\documentclass[reprint,pra,amssymb,amsmath,superscriptaddress,
longbibliography]{revtex4-1}

\pdfoutput=1

\usepackage{amstext}
\usepackage{amsthm}
\usepackage{color,hcolor}
\usepackage{graphicx}
\usepackage[colorlinks, citecolor=blue, linkcolor=blue,
urlcolor=blue]{hyperref}
\usepackage{amsfonts}\usepackage{mathrsfs}
\usepackage[normalem]{ulem}
\usepackage{mathptmx}

\newtheorem{theorem}{Theorem}[]
\newtheorem{proposition}[theorem]{Proposition}

\newcommand{\mrv}[1]{\mathbf{#1}}

\newcommand{\transp}[0]{\mathrm{T}}

\newcommand{\rel}[0]{\mathrm{rel}}
\newcommand{\secur}[0]{\mathrm{sec}}

\newcommand{\tr}{\operatorname{tr}}
\newcommand{\ket}[1]{\lvert{#1} \rangle}

\allowdisplaybreaks
\newcommand{\bi}{\mathbf{i}}

\newcommand{\bc}{\mathbf{c}}
\newcommand{\bk}{\mathbf{k}}

\newcommand{\bb}{\mathbf{b}}
\newcommand{\bx}{\mathbf{x}}

\newcommand{\bz}{\mathbf{z}}

\newcommand{\bxi}{{\boldsymbol{\xi}}}

\newcommand{\bs}{\mathbf{s}}
\newcommand{\Eve}{{\operatorname{Eve}}}

\newcommand{\Span}{{\operatorname{Span}}}

\begin{document}

\title{Security Against Collective Attacks of a Modified BB84 QKD
Protocol with Information only in One Basis}

\author{Michel Boyer}
\email{boyer@iro.umontreal.ca}
\affiliation{D\'epartement IRO, Universit\'e de Montr\'eal,
Montr\'eal (Qu\'ebec) H3C 3J7, Canada}
\author{Rotem Liss}
\email{rotemliss@cs.technion.ac.il}
\author{Tal Mor}
\email{talmo@cs.technion.ac.il}
\affiliation{Computer Science Department, Technion, Haifa 3200003,
Israel}

\begin{abstract}
The Quantum Key Distribution (QKD) protocol BB84 has been proven secure
against several important types of attacks: the collective attacks
and the joint attacks.
Here we analyze the security of a modified BB84
protocol, for which information is sent only in the $z$ basis
while testing is done in both the $z$ and the $x$ bases,
against collective attacks.
The proof follows the framework of a previous paper~\cite{BGM09},
but it avoids the classical information-theoretical
analysis that caused problems with composability.
We show that this modified BB84 protocol is as
secure against collective attacks as the original BB84 protocol,
and that it requires more bits for testing.
\end{abstract}

\maketitle

\section{Introduction}
Quantum Key Distribution (QKD) protocols take advantage of the laws
of quantum mechanics, and most of them can be proven secure
even against powerful adversaries limited only by the laws of physics.
The two parties (Alice and Bob) want to create a shared random key,
using an insecure quantum channel and an unjammable classical channel
(to which the adversary may listen, but not interfere).
The adversary (eavesdropper), Eve, tries to get as much information
as she can on the final shared key.
The first and most important QKD protocol is BB84~\cite{BB84}.

Boyer, Gelles, and Mor~\cite{BGM09} discussed the security of
the BB84 protocol against collective attacks.
Collective attacks~\cite{BM97a,BM97b,BBBGM02} are a subclass
of the joint attacks; joint attacks are the most powerful
theoretical attacks.
\cite{BGM09} improved the security proof of Biham, Boyer, Brassard,
van de Graaf, and Mor~\cite{BBBGM02}
against collective attacks,
by using some techniques of Biham, Boyer, Boykin, Mor,
and Roychowdhury~\cite{BBBMR06}
(that proved security against joint attacks).
In this paper, too, we restrict the analysis to collective attacks,
because security against collective attacks is conjectured
(and, in some security notions, proved~\cite{Renner05,MRR09})
to imply security against joint attacks. 
In addition, proving security against collective attacks is
much simpler than proving security against joint attacks.

In many QKD protocols, including BB84, Alice and Bob exchange
several types of bits (encoded as quantum systems, usually qubits):
INFO bits, that are secret bits shared by Alice and Bob and are used for
generating the final key (via classical processes of error correction
and privacy amplification); and TEST bits, that are publicly exposed
by Alice and Bob (by using the classical channel) and are used for
estimating the error rate. In BB84, each bit is sent from Alice
to Bob in a random basis (the $z$ basis or the $x$ basis).

In this paper, we extend the analysis of BB84 done in~\cite{BGM09}
and prove the security of a QKD protocol
we shall name \emph{BB84-INFO-$z$}.
This protocol is almost identical to BB84, except that all its INFO
bits are in the $z$ basis.
In other words, the $x$ basis is used only for testing. The bits
are thus partitioned into three disjoint sets: INFO, TEST-Z, and TEST-X.
The sizes of these sets are arbitrary ($n$ INFO bits, $n_z$
TEST-Z bits, and $n_x$ TEST-X bits).

We note that, while this paper follows a line of research that mainly
discusses a specific approach of security proof for BB84
and similar protocols (this approach,
notably, considers finite-key effects
and not only the asymptotic error rate),
many other approaches have also been suggested:
see for example~\cite{Mayers01,SP00,Renner05,bb84_sec_renner}.

In contrast to the line of research adopted here
(of~\cite{BM97a,BM97b,BBBGM02,BBBMR06,BGM09}),
in which a classical information-theoretical analysis caused
problems with composability (see definition in~\cite{Renner05}),
in this paper we suggest a method to avoid those problems:
we calculate the trace distance between any two density matrices Eve
may hold, instead of calculating the classical mutual information
between Eve and the final key (as done in those previous papers).
This method is implemented in this paper for the proof of BB84-INFO-$z$;
it also directly applies to the BB84 security proof in~\cite{BGM09},
and it may be extended in the future to show that the BB84
security proofs of~\cite{BGM09},~\cite{BBBGM02}, and~\cite{BBBMR06}
prove the composable security of BB84.

The ``qubit space'', $\mathscr{H}_2$, is a $2$-dimensional
Hilbert space. The states $\ket{0^0}, \ket{1^0}$ form an orthonormal
basis of $\mathscr{H}_2$, called ``the computational basis'' or
``the $z$ basis''. The states
$\ket{0^1} \triangleq \frac{1}{\sqrt{2}}[\ket{0^0} + \ket{1^0}]$ and
$\ket{1^1} \triangleq \frac{1}{\sqrt{2}}[\ket{0^0} - \ket{1^0}]$
form another orthonormal basis of $\mathscr{H}_2$,
called ``the $x$ basis''. Those two bases are said
to be \emph{conjugate bases}.

In this paper, bit strings of some length $t$ are denoted
by a bold letter (e.g., $\bi=i_1\ldots i_t$
with $i_1,\ldots,i_t \in \{0,1\}$)
and are identified to elements of the $t$-dimensional 
$\mathbf{F}_2$-vector space $\mathbf{F}_2^t$,
where $\mathbf{F}_2 = \{0,1\}$ and the addition of
two vectors corresponds to a XOR operation.
The number of $1$-bits in a bit string $\bs$ is denoted by $|\bs|$,
and the Hamming distance between two strings $\bs$ and $\bs'$
is $d_H(\bs, \bs') = |\bs + \bs'|$.

\section{\label{bb84-info-z_description}Formal Description
of the BB84-INFO-$z$ Protocol}
Below we describe the BB84-INFO-$z$ protocol used in this paper.
\begin{enumerate}
\item Alice and Bob pre-agree on numbers $n$, $n_z$, and $n_x$
(we denote $N \triangleq n + n_z + n_x$),
on error thresholds $p_{a,z}$ and $p_{a,x}$,
on a linear error-correcting code $C$ with
an $r\times n$ parity check matrix $P_C$,
and on a linear key-generation function (privacy amplification)
represented by an $m\times n$ matrix $P_K$.
It is required that \emph{all} the $r+m$ rows of the matrices
$P_C$ and $P_K$ put together are linearly independent.

\item Alice randomly chooses a partition
$\mathcal{P} = (\bs, \bz, \bb)$ of the $N$ bits by randomly choosing
three $N$-bit strings $\bs, \bz, \bb \in \mathbf{F}^{N}_2$
that satisfy $|\bs| = n, |\bz| = n_z, |\bb| = n_x$, and
$|\bs + \bz + \bb| = N$. $\mathcal{P}$ thus partitions the set
of indexes $\lbrace 1, 2, ..., N \rbrace$ into three disjoint sets:
\begin{itemize}
\item $I$ (INFO bits, where $s_j = 1$) of size $n$;
\item $T_Z$ (TEST-Z bits, where $z_j = 1$) of size $n_z$; and
\item $T_X$ (TEST-X bits, where $b_j = 1$) of size $n_x$.
\end{itemize}

\item Alice randomly chooses an $N$-bit string
$\bi \in \mathbf{F}^{N}_2$, and sends the $N$ qubit states
$\ket{i_1^{b_1}}, \ket{i_2^{b_2}}, \ldots, \ket{i_{N}^{b_{N}}}$,
one after the other, to Bob using the quantum channel.
Notice that the INFO and TEST-Z bits are encoded in the $z$ basis,
while the TEST-X bits are encoded in the $x$ basis.
Bob keeps each received qubit in quantum memory,
not measuring it yet~\footnote{
Here we assume that Bob has a quantum memory and can delay his
measurement. In practical implementations, Bob usually cannot do that,
but is assumed to measure in a randomly-chosen basis ($z$ or $x$),
so that Alice and Bob later discard the qubits measured in the wrong
basis. We assume that Alice sends more than $N$ qubits, so that $N$
qubits are finally detected by Bob and measured in the correct basis.}.

\item Alice publicly sends to Bob the string $\bb = b_1 \cdots b_N$.
Bob measures each saved qubit in the correct basis (namely,
if $b_i = 0$ then he measures the $i$-th qubit in the $z$ basis, and
if $b_i = 1$ then he measures it in the $x$ basis).

The bit string measured by Bob is denoted by $\bi^B$.
If there is no noise and no eavesdropping, then $\bi^B = \bi$.

\item Alice publicly sends to Bob the string $\bs$.
The INFO bits, used for generating the final key,
are the $n$ bits with $s_j = 1$,
while the TEST-Z and TEST-X bits are
the $n_z + n_x$ bits with $s_j = 0$.
The substrings of $\bi, \bb$ that correspond to the INFO bits
are denoted by $\bi_{\bs}$ and $\bb_{\bs}$.

\item Alice and Bob both publish their values of
all the TEST-Z and TEST-X bits, and compare the bit values.
If more than $n_z \cdot p_{a,z}$ of the TEST-Z bits are different
between Alice and Bob \emph{or} more than $n_x \cdot p_{a,x}$
of the TEST-X bits are different between them,
they abort the protocol. We note that $p_{a,z}$ and $p_{a,x}$
(the pre-agreed error thresholds) are the maximal allowed
error rates on the TEST-Z and TEST-X bits, respectively --
namely, in each basis ($z$ and $x$) separately.

\item Alice and Bob keep the values of the remaining
$n$ bits (the INFO bits, with $s_j = 1$) secret. 
The bit string of Alice is denoted $\bx = \bi_{\bs}$,
and the bit string of Bob is denoted $\bx^B$.

\item Alice sends to Bob the $r$-bit string
$\bxi = \bx P_C^\transp$, that is called the \emph{syndrome}
of $\bx$ (with respect to the error-correcting code $C$
and to its corresponding parity check matrix $P_C$).
By using $\bxi$, Bob corrects the errors in his $\bx^B$ string
(so that it is the same as $\bx$).

\item Alice and Bob compute the $m$-bit final key
${\bf k} = \bx P_K^\transp$.
\end{enumerate}

The protocol is defined similarly to BB84 (and to its description
in~\cite{BGM09}), except that it uses the generalized bit numbers
$n$, $n_z$, and $n_x$ (numbers of INFO, TEST-Z, and TEST-X bits,
respectively);
that it uses the partition $\mathcal{P} = (\bs, \bz, \bb)$ for
dividing the $N$-bit string $\bi$ into three disjoint sets of indexes
($I$, $T_Z$, and $T_X$); and that it uses two separate thresholds
($p_{a,z}$ and $p_{a,x}$) instead of one ($p_a$).

\section{Security Proof of {BB84-INFO-$z$} Against Collective Attacks}
\subsection{Results from~\cite{BGM09}}
The security proof of BB84-INFO-$z$ against collective attacks is very
similar to the security proof of BB84 itself against collective
attacks, that was detailed in~\cite{BGM09}. Most parts of the proof are
not affected at all by the changes made to BB84 to get the
BB84-INFO-$z$ protocol (changes detailed in
Section~\ref{bb84-info-z_description} of the current paper),
because those parts assume fixed strings $\bs$ and $\bb$,
and because the attack is collective (so the analysis is
restricted to the INFO bits).

Therefore, the reader is referred to the proof in Section~2 and
Subsections~3.1 to~3.5 of~\cite{BGM09}, that applies to BB84-INFO-$z$
without any changes (except changing the total number of bits,
$2n$, to $N$, which does not affect the proof at all),
and that will not be repeated here.

We denote the rows of the error-correction parity check matrix $P_C$
as the vectors $v_1,\ldots,v_r$ in $\mathbf{F}_2^n$,
and the rows of the privacy amplification matrix $P_K$
as the vectors $v_{r+1}, \ldots,v_{r+m}$.
We also define, for every $r'$,
$V_{r'} \triangleq \Span \lbrace v_1, ..., v_{r'} \rbrace$;
and we define
\begin{equation}
d_{r,m} \triangleq \min_{r
\leq r' < r+m} d_H(v_{r'+1}, V_{r'}) = \min_{r \leq r' < r+m}
d_{r',1}.
\end{equation}

For a $1$-bit final key $k \in \lbrace 0, 1 \rbrace$, we define
$\widehat{\rho}_k$ to be the state of Eve corresponding to the final
key $k$, given that she knows $\bxi$. Thus,
\begin{align}\label{rhohatk}
\widehat{\rho}_k &= \frac{1}{2^{n-r-1}}\sum_{\bx\, \big|
{\scriptsize\begin{matrix}\bx P_C^T = \bxi\\
\bx \cdot v_{r+1} = k\end{matrix}}} \rho_\bx^{\bb'},
\end{align}
where $\rho_\bx^{\bb'}$ is Eve's state after the attack, given that
Alice sent the INFO bits $\bx$ encoded in the bases $\bb' = \bb_\bs$.
We also defined in~\cite{BGM09} the state $\widetilde\rho_k$,
that is a lift-up of $\widehat{\rho}_k$
(which means that $\widehat{\rho}_k$ is a partial
trace of $\widetilde\rho_k$).

In the end of Subsection~3.5 of~\cite{BGM09}, it was found that
(in the case of a $1$-bit final key, i.e., $m=1$)
\begin{equation}\label{boundrho}
\frac{1}{2} \tr |\widetilde\rho_0 - \widetilde\rho_1| \leq 2
\sqrt{P\left[|\mathbf{C}_{I}| \geq \frac{d_{r,1}}{2} \ \mid
\ \mathbf{B}_{I} =
\overline{\bb'}, \bs \right]},
\end{equation}
where $\mathbf{C}_{I}$ is the random variable corresponding to
the $n$-bit string of errors on the $n$ INFO bits;
$\mathbf{B}_{I}$ is the random variable corresponding to
the $n$-bit string of bases of the $n$ INFO bits; $\overline{\bb'}$
is the bit-flipped string of $\bb' = \bb_\bs$; and $d_{r,1}$ (and,
in general, $d_{r,m}$) was defined above.

Now, according to \cite[Theorem~9.2 and page~407]{NCBook}, and using
the fact that $\widehat{\rho}_k$ is a partial trace of
$\widetilde\rho_k$, we find that
$\frac{1}{2} \tr |\widehat{\rho}_0 - \widehat{\rho}_1| \le
\frac{1}{2} \tr |\widetilde\rho_0 - \widetilde\rho_1|$. From this
result and from inequality \eqref{boundrho} we deduce that
\begin{equation}\label{boundrho2}
\frac{1}{2} \tr |\widehat{\rho}_0 - \widehat{\rho}_1| \leq 2
\sqrt{P\left[|\mathbf{C}_{I}| \geq \frac{d_{r,1}}{2} \ \mid
\ \mathbf{B}_{I} =
\overline{\bb'}, \bs \right]}.
\end{equation}

\subsection{Bounding the Differences Between Eve's States}
We define $\bc \triangleq \bi + \bi^B$: namely, $\bc$ is the XOR
of the $N$-bit string $\bi$ sent by Alice and
of the $N$-bit string $\bi^B$ measured by Bob.
For each index $1 \le l \le N$, $c_l=1$ if and only if
Bob's $l$-th bit value is different from the $l$-th bit sent by Alice.
The partition $\mathcal{P}$ divides the $N$ bits into
$n$ INFO bits, $n_z$ TEST-Z bits, and $n_x$ TEST-X bits.
The corresponding substrings of the error string $\bc$ are
$\bc_\bs$ (the string of errors on the INFO bits),
$\bc_\bz$ (the string of errors on the TEST-Z bits), and
$\bc_\bb$ (the string of errors on the TEST-X bits).
The random variables that correspond to
$\bc_\bs$, $\bc_\bz$, and $\bc_\bb$ are denoted by
$\mrv{C}_I$, $\mrv{C}_{T_Z}$, and $\mrv{C}_{T_X}$, respectively.

We define $\widetilde{\mrv{C}_I}$ to be the random variable
corresponding to the string of errors on the INFO bits
\emph{if Alice had encoded and sent the INFO bits in the $x$ basis}
(instead of the $z$ basis dictated by the protocol).
In those notations, inequality~\eqref{boundrho2} reads as
\begin{align}
\frac{1}{2} \tr |\widehat{\rho}_0 - \widehat{\rho}_1| &\leq 2
\sqrt{P\left[|\widetilde{\mathbf{C}_I}| \geq \frac{d_{r,1}}{2}
\ \mid\ \mathcal{P} \right]} \nonumber \\
&= 2
\sqrt{P\left[|\widetilde{\mathbf{C}_I}| \geq \frac{d_{r,1}}{2}
\ \mid\ \bc_\bz, \bc_\bb, \mathcal{P} \right]},\label{sdrewritten}
\end{align}
using the fact that Eve's attack is collective, so the qubits
are attacked independently, and, therefore, the errors on the INFO bits
are independent of the errors on the TEST-Z and TEST-X bits
(namely, of $\bc_\bz$ and $\bc_\bb$).

As described in~\cite{BGM09}, inequality~\eqref{sdrewritten}
was not derived for the actual attack
$U = U_1\otimes \ldots \otimes U_N$ applied
by Eve, but for a virtual flat attack (that depends on $\bb$
and therefore could not have been applied by Eve).
That flat attack gives the same states $\widehat{\rho}_0$
and $\widehat{\rho}_1$ as the original attack $U$, and gives a lower
(or the same) error rate in the conjugate basis. Therefore,
inequality~\eqref{sdrewritten} also holds for the original attack
$U$. This means that, from now on, all our results apply to
the original attack $U$ and not the flat attack.

So far, we have discussed a $1$-bit key.
We will now discuss a general $m$-bit key $\bk$.
We define $\widehat{\rho}_\bk$ to be the state of Eve corresponding
to the final key $\bk$, given that she knows $\bxi$:
\begin{align}
\widehat{\rho}_\bk &= \frac{1}{2^{n-r-m}}\sum_{\bx\,
\big|{\scriptsize\begin{matrix}\bx P_C^T =
\bxi\\ \bx P_K^T = \bk\end{matrix}}} \rho_\bx^{\bb'}
\end{align}

\begin{proposition}\label{probbound}
For any two $m$-bit keys $\bk,\bk'$,
\begin{align}
&\frac{1}{2} \tr |\widehat{\rho}_{\bk} - \widehat{\rho}_{\bk'}|
\nonumber \\
&\leq 2m\sqrt{
P\left[| \widetilde{\mrv{C}_I} | \geq \frac{d_{r,m}}{2}
\mid \bc_\bz, \bc_\bb, \mathcal{P}\right]}. \label{eqlemma}
\end{align}
\end{proposition}
\begin{proof}
We define the key $\bk_j$, for $0 \le j \le m$, to consist of the first
$j$ bits of $\bk'$ and the last $m-j$ bits of $\bk$. This means that
$\bk_0 = \bk$, $\bk_m = \bk'$, and $\bk_{j-1}$ differs from $\bk_j$
at most on a single bit (the $j$-th bit).

First, we find a bound on $\frac{1}{2} \tr |\widehat{\rho}_{\bk_{j-1}}
- \widehat{\rho}_{\bk_j}|$: since $\bk_{j-1}$ differs from $\bk_j$
at most on a single bit (the $j$-th bit, given by the formula
$\bx \cdot v_{r+j}$), we can use the same proof
that gave us inequality~\eqref{sdrewritten}, attaching the other
(identical) key bits to $\bxi$ of the original proof; and we find that:
\begin{align}
&\frac{1}{2} \tr |\widehat{\rho}_{\bk_{j-1}} - \widehat{\rho}_{\bk_j}|
\nonumber \\
&\leq 2 \sqrt{P\left[|\widetilde{\mathbf{C}_I}| \geq \frac{d_j}{2}
\ \mid\ \bc_\bz, \bc_\bb, \mathcal{P} \right]} \label{StatesDiff1}
\end{align}
where we define $d_j$ as $d_H(v_{r+j}, V'_j)$, and
$V'_j \triangleq \Span \lbrace v_1, v_2, \ldots, v_{r+j-1},
v_{r+j+1}, \ldots, v_{r+m} \rbrace$.

Now we notice that $d_j$ is the Hamming distance between $v_{r+j}$ and
some vector in $V'_j$, which means that $d_j = |\sum_{i =
1}^{r+m} a_i v_i|$ with $a_i \in \mathbf{F}_2$ and $a_{r+j} \ne 0$.
The properties of Hamming distance assure us that $d_j$ is
at least $d_H(v_{r'+1}, V_{r'})$ for some $r \le r' < r + m$.
Therefore, we find that $d_{r,m} = \min_{r \leq r' < r+m}
d_H(v_{r'+1}, V_{r'}) \le d_j$.

The result $d_{r,m} \le d_j$ implies that
if $|\widetilde{\mathbf{C}_I}| \geq \frac{d_j}{2}$
then $|\widetilde{\mathbf{C}_I}|\geq \frac{d_{r,m}}{2}$.
Therefore, inequality~\eqref{StatesDiff1} implies
\begin{align}
&\frac{1}{2} \tr |\widehat{\rho}_{\bk_{j-1}} - \widehat{\rho}_{\bk_j}|
\nonumber \\
&\leq 2 \sqrt{P\left[|\widetilde{\mathbf{C}_I}| \geq
\frac{d_{r,m}}{2} \ \mid \ \bc_\bz, \bc_\bb, \mathcal{P} \right]}.
\label{StatesDiff2}
\end{align}
Now we use the triangle inequality for norms to find
\begin{align}
&\frac{1}{2} \tr |\widehat{\rho}_{\bk} - \widehat{\rho}_{\bk'}|
\nonumber \\
&= \frac{1}{2} \tr |\widehat{\rho}_{\bk_0} - \widehat{\rho}_{\bk_m}|
\leq \sum_{j = 1}^m \frac{1}{2} \tr |\widehat{\rho}_{\bk_{j-1}}
- \widehat{\rho}_{\bk_j}| \nonumber \\
&\leq 2m \sqrt{P\left[|\widetilde{\mathbf{C}_I}|
\geq \frac{d_{r,m}}{2}
\ \mid\ \bc_\bz, \bc_\bb, \mathcal{P} \right]}.\label{StatesDiff3}
\end{align}
\end{proof}

The value we want to bound is the expected value of difference
between two states of Eve corresponding to two final keys. However, we
should take into account that if the test fails, no final key is
generated, and the difference between all of Eve's states becomes $0$
for any purpose. We thus define the random variable
$\Delta_\Eve^{(p_{a,z}, p_{a,x})}(\bk,\bk')$
for any two final keys $\bk,\bk'$:
\begin{align}
&\Delta_\Eve^{(p_{a,z}, p_{a,x})}(\bk, \bk' | \mathcal{P},\bxi,
\bc_\bz, \bc_\bb) \nonumber \\
&\triangleq
\begin{cases} 
\frac{1}{2} \tr |\widehat{\rho}_{\bk} - \widehat{\rho}_{\bk'}| 
& \text{if $\displaystyle \frac{|\bc_\bz|}{n_z} \leq p_{a,z}$ and
$\displaystyle \frac{|\bc_\bb|}{n_x} \leq p_{a,x}$} \\
0 & \text{otherwise} \end{cases}\label{defipa}
\end{align}

We need to bound the expected value
$\langle \Delta_\Eve^{(p_{a,z}, p_{a,x})}(\bk, \bk') \rangle$,
that is given by:
\begin{align}
\langle \Delta_\Eve^{(p_{a,z}, p_{a,x})}(\bk, \bk')\rangle
&= \sum_{\mathcal{P},\bxi, \bc_\bz, \bc_\bb} \nonumber \\
&\Delta_\Eve^{(p_{a,z}, p_{a,x})}(\bk, \bk' | \mathcal{P},\bxi,
\bc_\bz, \bc_\bb) \nonumber \\
&\cdot p(\mathcal{P},\bxi, \bc_\bz, \bc_\bb)\label{eqipa}
\end{align}
\begin{theorem}\label{thmsecurity}
\begin{align}
\langle \Delta_\Eve^{(p_{a,z}, p_{a,x})}(\bk, \bk')\rangle &\leq 2m
\sqrt{P\Big[\textstyle \left( \frac{| \widetilde{\mrv{C}_I}|}{n} \geq
\frac{d_{r,m}}{2n}\right)} \nonumber \\
&\overline{\textstyle \wedge \left(\frac{|\mrv{C}_{T_Z}|}{n_z} \leq
p_{a,z}\right)} \nonumber \\
&\overline{\textstyle \wedge \left(\frac{|\mrv{C}_{T_X}|}{n_x} \leq
p_{a,x}\right) \Big]} \label{boundotherbasis}
\end{align}
where $\frac {|\widetilde{\mrv{C}_I}|}{n}$ is the random variable
corresponding to the error rate on the INFO bits
if they had been encoded in the $x$ basis,
$\frac{|\mrv{C}_{T_Z}|}{n_z}$ is the random variable
corresponding to the error rate on the TEST-Z bits,
and $\frac{|\mrv{C}_{T_X}|}{n_x}$ is the random variable
corresponding to the error rate on the TEST-X bits.
\end{theorem}
\begin{proof}
We use the convexity of $x^2$, namely, the fact that for all $\{p_i\}_i$
satisfying $p_i \geq 0$ and $\sum_i p_i = 1$, it holds that
$(\sum_i p_i x_i)^2 \leq \sum_i p_i x_i^2$. We find that:
\begin{align*}
&\langle \Delta_\Eve^{(p_{a,z}, p_{a,x})}(\bk, \bk')\rangle^2 \\
= &
\Big[ \sum_{\mathcal{P},\bxi, \bc_\bz, \bc_\bb} \Delta_\Eve^{(p_{a,z},
p_{a,x})}(\bk, \bk' | \mathcal{P},\bxi, \bc_\bz, \bc_\bb) \\
& \cdot p( \mathcal{P}, \bxi, \bc_\bz, \bc_\bb)\Big]^2
\qquad\qquad\qquad \text{(by \eqref{eqipa})}\\
\leq &\sum_{\mathcal{P},\bxi, \bc_\bz, \bc_\bb} \left(
\Delta_\Eve^{(p_{a,z}, p_{a,x})}(\bk, \bk' | \mathcal{P},\bxi, \bc_\bz,
\bc_\bb) \right)^2 \\
&\cdot p(\mathcal{P},\bxi, \bc_\bz, \bc_\bb)
\qquad\qquad\qquad \text{(by convexity of $x^2$)}\\
= &
\sum_{\mathcal{P},\bxi,\frac{|\bc_\bz|}{n_z}\leq p_{a,z},
\frac{|\bc_\bb|}{n_x}\leq p_{a,x}} 
\left( \frac{1}{2} \tr |\widehat{\rho}_{\bk} -
\widehat{\rho}_{\bk'}| \right)^2 \\
&\cdot p( \mathcal{P}, \bxi, \bc_\bz, \bc_\bb)
\qquad\qquad\qquad \text{(by \eqref{defipa})} \\
\leq &
4m^2 \cdot \sum_{\mathcal{P},\bxi,\frac{|\bc_\bz|}{n_z}\leq p_{a,z},
\frac{|\bc_\bb|}{n_x}\leq p_{a,x}}
P\left[\textstyle | \widetilde{\mrv{C}_I}| \geq \frac{d_{r,m}}{2} \mid
\bc_\bz, \bc_\bb, \mathcal{P}\right] \\
&\cdot p(\mathcal{P},\bxi, \bc_\bz, \bc_\bb)
\qquad\qquad\qquad \text{(by \eqref{eqlemma})} \\
= &
4m^2 \cdot \sum_{\mathcal{P},\frac{|\bc_\bz|}{n_z}\leq p_{a,z},
\frac{|\bc_\bb|}{n_x}\leq p_{a,x}}
P\left[ \textstyle | \widetilde{\mrv{C}_I}| \geq \frac{d_{r,m}}{2}
\mid \bc_\bz, \bc_\bb, \mathcal{P}\right] \\
&\cdot p(\mathcal{P},\bc_\bz, \bc_\bb) \\
= &
4m^2 \cdot \sum_{\mathcal{P}}
P\Big[ \textstyle \left(| \widetilde{\mrv{C}_I} | \geq
\frac{d_{r,m}}{2}\right) \\
&\textstyle \wedge \left( \frac{|\mrv{C}_{T_Z}|}{n_z} \leq
p_{a,z}\right)
\wedge \left( \frac{|\mrv{C}_{T_X}|}{n_x} \leq p_{a,x}\right) \mid
\mathcal{P} \Big] \cdot p(\mathcal{P}) \nonumber \\ 
= & 4m^2 \cdot P\Big[ \textstyle \left(| \widetilde{\mrv{C}_I} | \geq
\frac{d_{r,m}}{2}\right) \\ 
&\textstyle \wedge \left( \frac{|\mrv{C}_{T_Z}|}{n_z} \leq
p_{a,z}\right)
\wedge \left( \frac{|\mrv{C}_{T_X}|}{n_x} \leq p_{a,x}\right)\Big]
\end{align*}
\end{proof}

\subsection{\label{sec_proof}Proof of Security}
Following~\cite{BGM09} and~\cite{BBBMR06},
we choose matrices $P_C$ and $P_K$ such that the inequality
$\frac{d_{r,m}}{2n} > p_{a,x} + \epsilon$
is satisfied for some $\epsilon$ (we will explain in
Subsection~\ref{sec_rel_rate_thres} why this is possible).
This means that
\begin{align}\label{PcPkInequality}
& P\left[\textstyle \left( \frac{| \widetilde{\mrv{C}_I}|}{n} \geq
\frac{d_{r,m}}{2n}\right) \wedge
\left(\frac{|\mrv{C}_{T_Z}|}{n_z} \leq p_{a,z}\right) \wedge
\left(\frac{|\mrv{C}_{T_X}|}{n_x} \leq p_{a,x}\right) \right]
\nonumber \\
&\leq P\left[\textstyle \left(\frac{|\widetilde{\mrv{C}_I}|}{n} >
p_{a,x} + \epsilon\right) \wedge \left(\frac{|\mrv{C}_{T_X}|}{n_x} \leq
p_{a,x}\right) \right].
\end{align}
We will now prove the right-hand-side of~\eqref{PcPkInequality}
to be exponentially small in $n$.

As said earlier, the random variable $\widetilde{\mrv{C}_I}$
corresponds to the bit string of errors on the INFO bits
if they had been encoded in the $x$ basis. The TEST-X bits are
also encoded in the $x$ basis, and the random variable $\mrv{C}_{T_X}$
corresponds to the bit string of errors on those bits.
Therefore, we can treat the selection of the $n$ INFO bits
and of the $n_x$ TEST-X bits as a random sampling (after the numbers
$n$, $n_z$, and $n_x$ \emph{and} the TEST-Z bits have all already
been chosen), and use Hoeffding's theorem
(that is described in Appendix~A of~\cite{BGM09}).

Therefore, for each bit string $c_1\ldots c_{n+n_x}$ that consists
of the errors in the $n+n_x$ INFO and TEST-X bits
\emph{if the INFO bits had been encoded in the $x$ basis},
we apply Hoeffding's theorem: namely, we take a sample of size $n$
without replacement from the population $c_1, \ldots, c_{n+n_x}$
(this corresponds to the random selection of the INFO bits and
the TEST-X bits, as defined above, given that the TEST-Z bits
have already been chosen).
Let $\overline{X} = \frac{|\widetilde{\mrv{C}_I}|}{n}$
be the average of the sample (this is exactly the error rate
on the INFO bits, assuming, again, the INFO bits had been encoded
in the $x$ basis); and let
$\mu = \frac{|\widetilde{\mrv{C}_I}| + |\mrv{C}_{T_X}|}{n + n_x}$ 
be the expectancy of $\overline{X}$ (this is exactly the error rate
on the INFO bits and TEST-X bits together).
Then $\frac{|\mrv{C}_{T_X}|}{n_x} \leq p_{a,x}$ is equivalent to
$(n + n_x)\mu - n\overline{X} \leq n_x \cdot p_{a,x}$,
and, therefore, to $n \cdot (\overline{X}-\mu)
\geq n_x \cdot (\mu - p_{a,x})$.
This means that the conditions
$\left(\frac{|\widetilde{\mrv{C}_I}|}{n} > p_{a,x}+\epsilon\right)$ and
$\left(\frac{|\mrv{C}_{T_X}|}{n_x} \leq p_{a,x}\right)$ rewrite to
\begin{align}
&\left(\overline{X} - \mu > \epsilon + p_{a,x} - \mu\right) \nonumber \\
&\wedge \left(\frac{n}{n_x} \cdot (\overline{X}-\mu)
\geq \mu - p_{a,x}\right), \label{inequalities}
\end{align}
which implies $\left(1 + \frac{n}{n_x}\right)(\overline{X} -\mu) >
\epsilon$, which is equivalent to $\overline{X} -\mu >
\frac{n_x}{n + n_x} \epsilon$.
Using Hoeffding's theorem (from Appendix~A of~\cite{BGM09}), we get:
\begin{align}
& P\left[ \left(\frac{|\widetilde{\mrv{C}_I}|}{n} > p_{a,x} +
\epsilon\right) \wedge \left(\frac{|\mrv{C}_{T_X}|}{n_x} \leq
p_{a,x}\right) \right] \nonumber \\
&\leq P\left[ \overline{X} - \mu > \frac{n_x}{n + n_x}\epsilon\right]
\leq e^{-2 \left(\textstyle\frac{n_x}{n + n_x}\right)^2 n\epsilon^2}
\end{align}

In the above discussion, we have actually proved the following Theorem:
\begin{theorem}\label{thmsecurity1}
Let us be given $\delta > 0$, $R > 0$, and,
for infinitely many values of $n$,
a family $\{v^n_1, \ldots, v^n_{r_n+m_n}\}$ of
linearly independent vectors in $\mathbf{F}_2^n$ such that
$\delta < \frac{d_{r_n,m_n} }{n}$ and $\frac{m_n}{n} \leq R$.
Then for any $p_{a,z}, p_{a,x} > 0$ and $\epsilon_\secur > 0$
such that $p_{a,x} + \epsilon_\secur \leq \frac{\delta}{2}$,
and for any $n, n_z, n_x > 0$ and two $m_n$-bit final keys
$\bk,\bk'$, Eve's difference between her states
corresponding to $\bk$ and $\bk'$ satisfies the following bound:
\begin{align}\label{EveInfoBound2}
\langle \Delta_\Eve^{(p_{a,z}, p_{a,x})}(\bk, \bk')\rangle \leq 2R\,
n e^{- \left(\textstyle\frac{n_x}{n + n_x}\right)^2 n \epsilon^2_\secur}
\end{align}
\end{theorem}
In Subsection~\ref{sec_rel_rate_thres} we explain why
this Theorem guarantees security.

We note that the quantity
$\langle \Delta_\Eve^{(p_{a,z}, p_{a,x})}(\bk, \bk')\rangle$ bounds
the expected values of the Shannon Distinguishability and of the
mutual information between Eve and the final key,
as done in~\cite{BGM09} and~\cite{BBBMR06},
which is sufficient for proving non-composable security;
but it also avoids composability problems: Eve is not required to
measure immediately after the protocol ends, but she is allowed to wait
until she gets more information; and equation~\eqref{EveInfoBound2}
bounds the trace distance between any two of Eve's possible states.

\subsection{\label{rel_proof}Reliability}
Security itself is not sufficient; we also need the key to be reliable
(namely, to be the same for Alice and Bob). This means that we should
make sure that the number of errors on the INFO bits is less than
the maximal number of errors that can be corrected by the
error-correcting code. We demand that our error-correcting code can
correct $n(p_{a,z}+\epsilon_\rel)$ errors. Therefore, reliability
of the final key with exponentially small probability of failure
is guaranteed by the following inequality:
(as said, $\mrv{C}_I$ corresponds to the actual bit string of errors
on the INFO bits in the protocol, when they are encoded
in the $z$ basis)
\begin{align*}
& P\left[ \left(\frac{|\mrv{C}_I|}{n} > p_{a,z} + \epsilon_\rel\right)
\wedge \left(\frac{|\mrv{C}_{T_Z}|}{n_z} \leq p_{a,z}\right) \right] \\
&\leq e^{-2 \left(\textstyle\frac{n_z}{n + n_z}\right)^2
n\epsilon_\text{rel}^2}
\end{align*}
This inequality is proved by an argument similar to the one used in
Subsection~\ref{sec_proof}: the selection of the INFO bits
and TEST-Z bits is a random partition of $n + n_z$ bits into
two subsets of sizes $n$ and $n_z$, respectively (assuming that the
TEST-X bits have already been chosen), and thus it corresponds
to Hoeffding's sampling.

\subsection{\label{sec_rel_rate_thres} Security, Reliability,
and Error Rate Threshold}
According to Theorem~\ref{thmsecurity1} and to the discussion in
Subsection~\ref{rel_proof}, to get both security and reliability
we only need vectors $\{v^n_1, \ldots, v^n_{r_n+m_n}\}$
satisfying both the conditions of the Theorem (distance
$\frac{d_{r_n,m_n}}{2n} > \frac{\delta}{2}
\geq p_{a,x} + \epsilon_\secur$)
and the reliability condition
(the ability to correct $n(p_{a,z}+\epsilon_\rel)$ errors).
Such families were proven to exist in Appendix~E of~\cite{BBBMR06},
giving the bit-rate:
\begin{eqnarray}
R_{\mathrm{secret}} &\triangleq& \frac{m}{n} \nonumber \\
&=& 1 - H_2(2 p_{a,x} + 2 \epsilon_\secur) \nonumber \\
&-& H_2 \left( p_{a,z} + \epsilon_\rel + \frac{1}{n} \right)
\end{eqnarray}
where $H_2(x) \triangleq -x \log_2(x) -(1-x) \log_2(1-x)$.

Note that we use here the error thresholds $p_{a,x}$ for security and
$p_{a,z}$ for reliability. This is possible,
because in~\cite{BBBMR06} those conditions (security and reliability)
on the codes are discussed separately.

To get the asymptotic error rate thresholds, we require
$R_{\mathrm{secret}} > 0$, and we get the condition:
\begin{equation}
H_2(2 p_{a,x} + 2 \epsilon_\secur) + H_2 \left( p_{a,z}
+ \epsilon_\rel + \frac{1}{n} \right) < 1
\end{equation}

The secure asymptotic error rate thresholds zone is shown in
Figure~\ref{fig:security} (it is below the curve),
assuming that $\frac{1}{n}$
is negligible. Note the trade-off between the error rates $p_{a,z}$ and
$p_{a,x}$. Also note that in the case $p_{a,z} = p_{a,x}$, we get the
same threshold as BB84 (\cite{BBBMR06} and~\cite{BGM09}), which is
7.56\%.

\begin{figure}
\includegraphics[width=\linewidth]{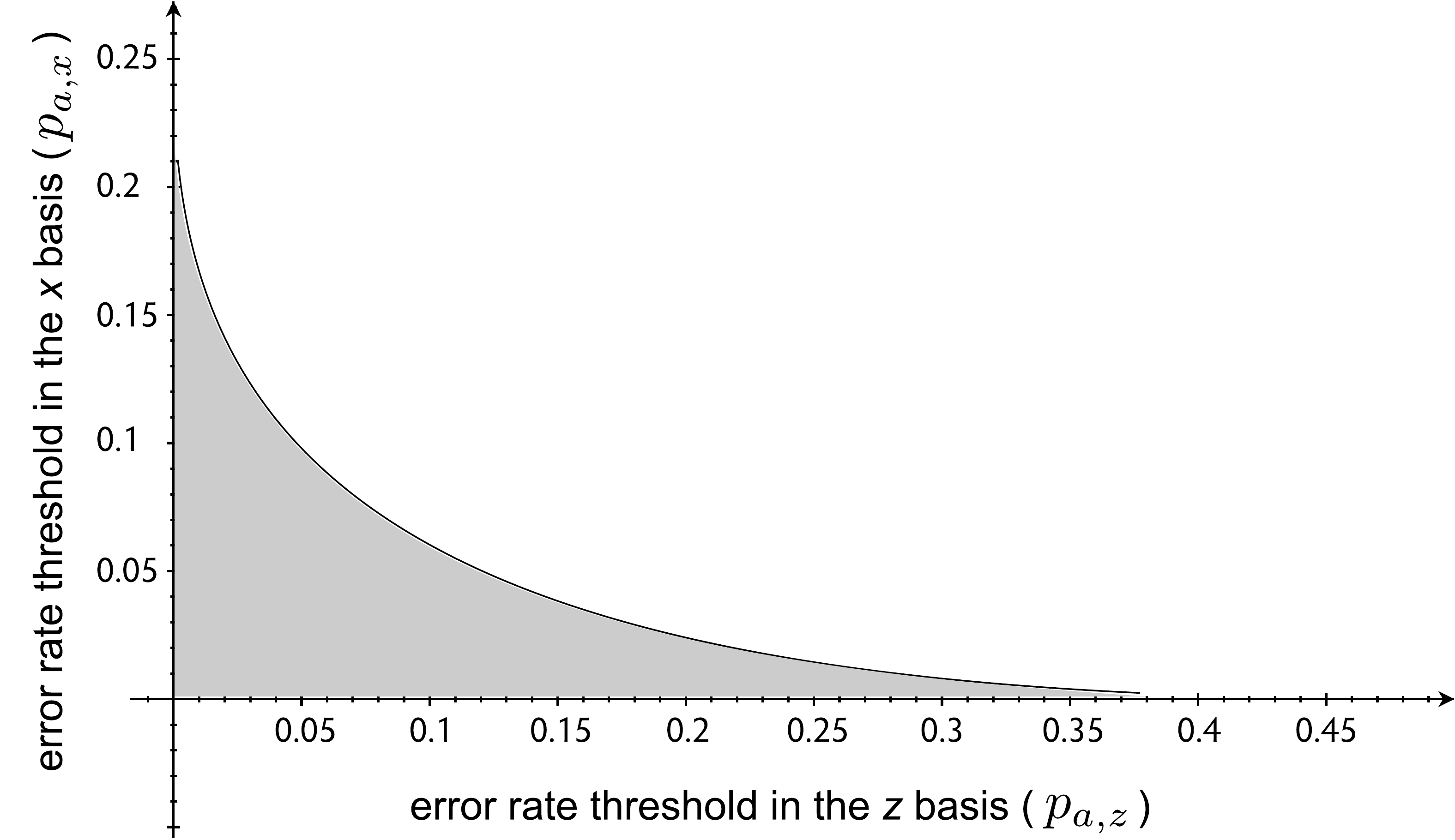}
\caption{\label{fig:security}The secure asymptotic error rates
zone (below the curve)}
\end{figure}

\section{Conclusion}
In this paper, we have analyzed the security of the BB84-INFO-$z$
protocol against any collective attack. We have discovered that
the results of BB84 hold very similarly for BB84-INFO-$z$,
with only two exceptions:
\begin{enumerate}
\item The error rates must be \uline{separately} checked to be
below the thresholds $p_{a,z}$ and $p_{a,x}$ for the TEST-Z and
TEST-X bits, respectively, while in BB84 the
error rate threshold $p_a$ applies to all the TEST bits
together.
\item The exponents of Eve's information (security) and of the
failure probability of the error-correcting code (reliability)
are different than in~\cite{BGM09}, because different numbers
of test bits are now allowed ($n_z$ and $n_x$ are arbitrary).
This implies that the exponents may decrease more slowly
(or more quickly) as a function of $n$.
However, if we choose $n_z = n_x = n$ (thus sending $N = 3n$
qubits from Alice to Bob), then we get exactly the same exponents
as in~\cite{BGM09}.
\end{enumerate}

The asymptotic error rate thresholds found in this paper are more
flexible than in BB84, because they allow us to tolerate a higher
threshold for a specific basis (say, the $x$ basis) if we demand a
lower threshold for the other basis ($z$). If we choose the same
error rate threshold for both bases, then the asymptotic bound is
7.56\%, exactly the bound found for BB84
in~\cite{BBBMR06} and~\cite{BGM09}.

We conclude that even if we change the BB84 protocol to have INFO bits
only in the $z$ basis, this does not harm its security and reliability
(at least against collective attacks).
This does not even change the asymptotic error rate threshold,
and allows more flexibility when choosing the thresholds for both bases.
The only drawbacks of this change are the need to check the
error rate for the two bases separately,
and the need to either send more qubits
($3n$ qubits in total, rather than $2n$)
or get a slower exponential decrease of the exponents required
for security and reliability.

We thus find that the feature of BB84, that both bases are used for
information, is not very important for security and reliability,
and that BB84-INFO-$z$ (that lacks this feature)
is almost as useful as BB84.
This may have important implications on the security and reliability of
other protocols that also only use one basis for information qubits,
as done in some two-way protocols.

We also present a better approach for the proof, that uses a quantum
distance between two states rather than the classical information.
In~\cite{BGM09},~\cite{BBBGM02}, and~\cite{BBBMR06},
the classical mutual information
between Eve's information (after an optimal measurement)
and the final key was calculated
(by using the trace distance between two quantum states);
although we should note that in~\cite{BGM09} and~\cite{BBBMR06},
the trace distance
was used for the proof of security of a single bit of the final key
even when all other bits are given to Eve, and only the last stages
of the proof discussed bounding the classical mutual information.
In the current paper, on the other hand, we use the trace distance
between the two quantum states until the end of the proof,
which avoids composability problems that existed in the previous works.

Therefore, this proof makes a step towards
making~\cite{BGM09},~\cite{BBBGM02},
and~\cite{BBBMR06} prove composable security of BB84
(namely, security even if Eve keeps her quantum states
until she gets more information when Alice and Bob use the key,
rather than measuring them in the end of the protocol).
This approach also applies (similarly) to the BB84 security proof
in~\cite{BGM09}.

\begin{acknowledgments}
The work of TM and RL was partly supported
by the Israeli MOD Research and Technology Unit.
\end{acknowledgments}

\bibliography{security}

\end{document}